\documentclass[letterpaper, 10 pt, conference]{ieeeconf}  

\usepackage{amsmath,amssymb,amsfonts}
\usepackage{graphicx}
\usepackage{textcomp}
\usepackage{xcolor}
\def\BibTeX{{\rm B\kern-.05em{\sc i\kern-.025em b}\kern-.08em    T\kern-.1667em\lower.7ex\hbox{E}\kern-.125emX}}
\usepackage{verbatim}
\usepackage{paralist}
\usepackage{graphics} 
\usepackage{epsfig} 
\usepackage{mathptmx} 
\usepackage{times} 
\usepackage{physics}
\usepackage{mdwmath}
\usepackage{mdwtab}
\usepackage[hidelinks]{hyperref}
\usepackage{algpseudocode}
\usepackage{algorithm}
\usepackage{caption}
\usepackage{subcaption}
\usepackage{cite}

\newtheorem{remark}{Remark}
\usepackage{tikz}
\usetikzlibrary{shapes, arrows.meta, positioning}
\newcommand{\argmin}{\mathop{\mathrm{arg\,min}}}

\newtheorem{theorem}{Theorem}%
\newtheorem{proposition}[theorem]{Proposition}%

\IEEEoverridecommandlockouts                              
\title{\LARGE \bf
Optimizing Maximally Entangled State Generation \\ via Pontryagin’s Principle
}

\author{Nahid Binandeh Dehaghani, A. Pedro Aguiar, Rafal Wisniewski
\thanks{N. Dehaghani and R. Wisniewski are with Department of Electronic Systems, Aalborg University, Fredrik Bajers vej 7c, DK-9220 Aalborg, Denmark
 {\tt\small \{nahidbd,raf\}@es.aau.dk}}
\thanks{A. Aguiar is with the Research Center for Systems and Technologies (SYSTEC), Electrical and Computer Engineering Department, FEUP - Faculty of Engineering, University of Porto, Rua Dr. Roberto Frias sn, i219, 4200-465 Porto, Portugal
        {\tt\small pedro.aguiar@fe.up.pt}}%
}

\begin{document}

\maketitle
\thispagestyle{empty}
\pagestyle{empty}

\begin{abstract}
We propose an optimal control strategy to generate maximally entangled states in bipartite quantum systems. Leveraging the Pontryagin Principle, we derive time-dependent control fields that maximize the entanglement measure, specifically concurrence, within minimal time while adhering to input constraints. Our formulation addresses the Liouville-von Neumann dynamics of the reduced density matrix under unitary evolution. The strategy is numerically validated through simulations, demonstrating its ability to drive the system from an initial perturbed separable state to a maximally entangled target state. The results showcase the effectiveness of switching control fields in optimizing entanglement, with potential applications in quantum technologies, including communication, computation, and sensing.

\end{abstract}

\section{Introduction}
Quantum entanglement is a cornerstone of quantum mechanics, enabling non-classical correlations between subsystems that contribute to many quantum technologies \cite{horodecki2009quantum}. It serves as a resource for quantum computation, communication, and sensing, where the ability to generate and utilize entangled states is critical. 
The preparation of high-fidelity entangled states requires precise control techniques to address the complexities of quantum systems, ensuring their effective utilization in diverse applications.

While several approaches to entanglement generation have been explored \cite{Lee2022GenerationOM, kues2017chip, riera2023remotely, gonzalez2015generation, leifer2003optimal}, the potential of optimal control theory \cite{boscain2021introduction, koch2022quantum} for generating maximally entangled states has yet to be explored. Quantum optimal control \cite{dehaghani2023quantum, dehaghani2022quantum} is promising to offer a powerful framework to address the challenges associated with entanglement generation, enabling precise manipulation of quantum systems to achieve specific target states or dynamics. By designing tailored control fields, this approach can enhance system performance and facilitate the preparation of entangled states. Furthermore, by formulating control problems that maximize performance metrics—such as entanglement fidelity—while minimizing resource consumption, such as time or energy, optimal control theory provides a pathway to address limitations in quantum technologies.


Motivated by the aforementioned challenges, we propose a quantum optimal control strategy tailored for the preparation of entangled states. Our approach formulates and solves a control problem aimed at maximizing entanglement within the shortest possible time, while adhering to input constraints. To demonstrate its effectiveness, we apply this strategy to a closed quantum system governed by the Liouville-von Neumann equation, highlighting its capability to prepare entangled states under unitary dynamics.

The remainder of the work is organized as follows: Section~\ref{sec:preliminaries} provides the necessary preliminaries, including a detailed discussion on entanglement monotones and measures, with a specific focus on concurrence analysis and the Pontryagin Minimum Principle (PMP) for optimal control. Section~\ref{sec:optimal_control} introduces the optimal control framework based on quantum entanglement measures, including the problem formulation and necessary conditions derived from Pontryagin's Minimum Principle. Finally, in Section~\ref{sec:numerical_validation}, we validate the proposed control strategies through numerical simulations and conclude with a summary of findings and future directions.\\

\noindent\textbf{Notation:} For a general continuous-time trajectory $x$, the term $x(t)$ indicates the trajectory evaluated at a specific time $t$. For writing the conjugate transpose of a matrix, we use the superscript $\dagger$. To denote wave functions as vectors, we use the Dirac notation such that $\left| \psi  \right\rangle =\sum\limits_{k=1}^{n}{{{\alpha }_{k}}\left| {{{{\psi }}}_{k}} \right\rangle }$, where $\vert \psi \rangle $ indicates a state vector, ${\alpha }_{k}$ are the complex-valued expansion coefficients, and $\vert {\psi}_k \rangle $ are fixed basis vectors. The notation bra is defined such that $\left \langle \psi  \right |= \left| \psi  \right\rangle^\dagger$.
The quantum density operator is denoted by 
$\rho =\sum\limits_{j}{{{p}_{j}}\left| {{\psi }_{j}} \right\rangle \left\langle  {{\psi }_{j}} \right|}$ where coefficients 
${p_{j}}$ are non-negative and add up to one.
For writing partial differential equations, we denote partial derivatives using $\partial$. The sign $\otimes$ indicates the tensor product. The notation $[\cdot,\cdot]$ represents a commutator. We denote the trace of a square matrix $A$ by $\operatorname{Tr}(A)$, and the partial trace over subsystem $B$ by $\operatorname{Tr}_B$.
Throughout the paper, the imaginary unit is $i = \sqrt{-1}$.
A d-dimensional Hilbert space is represented by $\mathbb{H}_d$,
and $\mathcal{H}_\rho$ is the set of linear bounded operators on 
$\mathbb{H}_d$ characterized by $d \times d$ complex matrices that are positive semi-definite, Hermitian, and have trace one. The set of complex numbers and purely imaginary numbers are shown by $\mathbb{C}$ and $i\mathbb{R}$, respectively.

\section{Preliminaries}\label{sec:preliminaries}
In this section, we provide the foundational concepts and theoretical background necessary for understanding the contributions of this work. We begin with the notions of entanglement monotones and measures. 
We then discuss the specific use of concurrence as a tool for analyzing quantum entanglement, and end with 
an overview of optimal control theory and necessary optimality conditions.

\subsection{Entanglement Monotones and Measures}
In this subsection, we introduce fundamental concepts from entanglement theory, focusing on entanglement monotones and measures. 

Entanglement measures must satisfy specific criteria to accurately reflect the entanglement properties of quantum states. For a bipartite entanglement measure \( E(\rho) \), the following postulates are generally accepted as defining characteristics \cite{donald2002uniqueness, vedral1998entanglement, vidal2000entanglement, horodecki2001entanglement}:
\begin{enumerate}
    \item {Quantification of Entanglement}: \( E(\rho) \) quantifies the intensity of entanglement of a state \( \rho \) by mapping system density matrices to positive real numbers, i.e., \( \rho \mapsto E(\rho) \in \mathbb{R}^{+} \). 

    \item {Zero for Separable States}: \( E(\rho) = 0 \) if and only if \( \rho \) is separable, indicating the absence of quantum entanglement. 

    \item {Monotonicity under LOCC}: \( E(\rho) \) does not increase on average under Local Operations and Classical Communication (LOCC) operations. For any LOCC operation \( \Lambda \), it holds that \( E(\Lambda \rho) \le E(\rho) \), reflecting the principle that entanglement cannot be increased through local actions and classical communication alone.

    \item {Agreement with Entropy of Entanglement for Pure States}: For pure states \( \rho = \left| \psi \right\rangle \left\langle \psi \right| \), the measure converges to the entropy of entanglement, aligning the entanglement measure with the von Neumann entropy for pure bipartite systems.
\end{enumerate}

Any function \( E(\rho) \) that satisfies the first three conditions is classified as an entanglement monotone. However, functions that fulfill conditions 1, 2, and 4, while also remaining non-increasing under deterministic LOCC transformations, are regarded as entanglement measures \cite{liintroduction}.
Some authors further define entanglement measures by imposing additional criteria, such as convexity, additivity, and continuity, as detailed in \cite{donald2002uniqueness, vedral1998entanglement}.

While $E(\rho)$ represents a general entanglement measure, different specific functions have been proposed to quantify entanglement. Key entanglement quantification criteria include the entropy of entanglement, Renyi entropy, entanglement of formation, Negativity, and Concurrence.
Amonst them, Concurrence, represented by $C(\rho)$, serves as an entanglement monotone and provides a compact way to quantify entanglement, particularly for bipartite qubit systems. In certain cases, entanglement measures, such as entanglement of formation, can also be expressed in terms of concurrence \cite{wootters2001entanglement}.
Given the importance and widespread use of concurrence as an entanglement monotone, our work will primarily focus on analyzing quantum entanglement through the lens of concurrence, as it provides a practical and insightful tool for quantifying entanglement.

\subsubsection{Analysis of Quantum Entanglement Using Concurrence}
For a bipartite quantum system described by a pure state \( |\psi\rangle \) in the Hilbert space \( \mathbb{H}_{d_A} \otimes \mathbb{H}_{d_B} \), the state can be expressed in the computational basis \( \{|ij\rangle\} \) as 
\begin{equation}
|\psi\rangle = \sum_{i=1}^{d_A} \sum_{j=1}^{d_B} \alpha_{ij} |i\rangle_A |j\rangle_B,
\end{equation}
where \( \alpha_{ij} \in \mathbb{C} \) are complex coefficients and satisfy the normalization condition 
\begin{equation}
\sum_{i=1}^{d_A} \sum_{j=1}^{d_B} |\alpha_{ij}|^2 = 1.
\end{equation}
The concurrence \( C(\rho) \) provides a compact way to quantify the entanglement of the state \( |\psi\rangle \). For separable states, 
the concurrence is zero. Conversely, for maximally entangled states, 
the maximum concurrence value
achieves its maximum value of 1. This corresponds to maximally entangled states such as the Bell states. 
The degree of entanglement is quantified as
\begin{equation}
C(\rho) = \sqrt{2 \left(1 - \operatorname{Tr}(\rho_A^2)\right)},
\end{equation}
where \( \rho_A \) is the reduced density matrix of subsystem \( A \), obtained by tracing out the degrees of freedom of subsystem \( B \), that is,
\begin{equation}
\rho_A = \operatorname{Tr}_B(|\psi\rangle \langle\psi|).
\end{equation}
To analyze the structure of \( |\psi\rangle \), we use the Schmidt decomposition, which expresses the pure state as \cite{Lee2022GenerationOM, nielsen1999conditions}
\begin{equation}
|\psi\rangle = \sum_{k=1}^r \sqrt{\lambda_k} \, |u_k\rangle_A \otimes |v_k\rangle_B,
\end{equation}
where
$\sqrt{\lambda_k}$ are the Schmidt coefficients, with \( \lambda_k \geq 0 \) and \( \sum_{k=1}^r \lambda_k = 1 \),
     \( |u_k\rangle_A \) and \( |v_k\rangle_B \) are orthonormal basis sets for \( \mathbb{H}_{d_A} \) and \( \mathbb{H}_{d_B} \), respectively, and
     \( r \) is the Schmidt rank.  
The Schmidt decomposition highlights the shared entanglement structure between the two subsystems.
Using the Schmidt decomposition, the reduced density matrices of subsystems \( A \) and \( B \) can be written in their spectral decompositions
\begin{equation}\label{schmit}
\rho_A \triangleq \operatorname{Tr}_B(\rho_{AB}) = \sum_{k=1}^r \left( \mathbb{I}_A \otimes \langle v_k^B | \right) \rho_{AB} \left( \mathbb{I}_A \otimes |v_k^B \rangle \right)
\end{equation}
where \( \mathbb{I}_A \) is the identity matrix in \( \mathbb{H}_{d_A} \). 
These reduced density matrices encapsulate the bipartite system's entanglement, with their eigenvalues and eigenvectors derived directly from the Schmidt decomposition.


\subsection{Optimal Control Problems via Minimum Principle of Pontryagin }\label{PMP}
We briefly revisit the formulation of optimal control problems, following the approach presented in \cite{pontryagin2018mathematical,bryson2018applied}, and extend its application to the quantum context in the subsequent section.
In a generic optimal control problem, we aim to minimize a given cost function, which is typically expressed as
\begin{equation}
{J} = \Phi\big(\mathbf{x}(t_0), t_0, {x}(t_f), t_f\big) + \int_{t_0}^{t_f} \mathcal{L} \big(\mathbf{x}(t), {u}(t), t\big) \, dt
\label{eq:cost_function}
\end{equation}
Here, the dynamics of the system are governed by the following differential equation
\begin{equation}
\dot{\mathbf{x}} = f\big(\mathbf{x}(t), \mathbf{u}(t), t\big),
\label{eq:state_dynamics}
\end{equation}
subject to the boundary conditions:
\begin{equation}
\Phi\big(\mathbf{x}(t_0), t_0\big) = \Phi_0,
\label{eq:initial_boundary}
\end{equation}
\begin{equation}
\Phi\big(\mathbf{x}(t_f), t_f\big) = \Phi_f.
\label{eq:final_boundary}
\end{equation}
In this framework, $\mathbf{x}(t)$ represents the state vector of the system \eqref{eq:state_dynamics}, $u(t)$ is the control vector, and $t$ denotes the independent time variable. The initial and final times are $t_0$ and $t_f$, respectively. The function $\Phi$ corresponds to the endpoint cost, and $\mathcal{L}$ denotes the running cost (Lagrangian).

To solve the optimization problem using the Pontryagin Minimum Principle (PMP), we introduce the Hamiltonian function \cite{Primerross}
\begin{equation}
\mathcal{H} = \mathcal{L} + {\pi}^T f,
\label{eq:hamiltonian}
\end{equation}
where ${\pi}$ represents the adjoint (costate) variable. According to PMP, the optimal control ${u}^*(t)$ is determined by solving the first-order optimality condition:
\begin{equation}
u^\star=\argmin_{u(t)} \mathcal{H}(x(t), u(t), \pi(t))
\label{eq:optimal_control}
\end{equation}
Once the control ${u}^*(t)$ is obtained from~\eqref{eq:optimal_control}, it is substituted back into the Hamiltonian. This enables us to derive the first-order conditions for the evolution of the state and adjoint variables as follows:
\begin{equation}
\dot{\mathbf{x}} = \frac{\partial \mathcal{H}}{\partial{\pi}},
\label{eq:state_evolution}
\end{equation}
\begin{equation}
\dot{{\pi}} = - \frac{\partial \mathcal{H}}{\partial {x}}.
\label{eq:adjoint_evolution}
\end{equation}
Additionally, if certain state variables are free at the initial or final times, specific transversality conditions need to be satisfied. These conditions are given as:
\begin{equation}
\boldsymbol{\pi}(t_0) = - \frac{\partial \mathcal{J}}{\partial \mathbf{x}_0},
\label{eq:transversality_initial_costate}
\end{equation}

\begin{equation}
H(t_0) = \frac{\partial {J}}{\partial t_0},
\label{eq:transversality_initial_time}
\end{equation}

\begin{equation}
\boldsymbol{\pi}(t_f) = \frac{\partial {J}}{\partial {x}_f},
\label{eq:transversality_final_costate}
\end{equation}

\begin{equation}
H(t_f) = -\frac{\partial{J}}{\partial t_f},
\label{eq:transversality_final_time}
\end{equation}
Equations \eqref{eq:state_evolution} and \eqref{eq:adjoint_evolution}, combined with the transversality conditions on the Hamiltonian, constitute a boundary value problem. 


\section{Optimal Control Based on Quantum Entanglement Measures }\label{sec:optimal_control}
This section examines the application of optimal control strategies to enhance the utility of entangled states. We begin by formulating the optimal control problem, followed by presenting the necessary optimality conditions based on Pontryagin’s Minimum Principle, building on the discussion provided earlier in Section \ref{PMP}.

\subsection{Optimal Control Problem Statement}
Consider a bipartite quantum system composed of two interacting particles. The objective is to manipulate the system's state to achieve a high degree of entanglement. To achieve this, we apply Pontryagin's Minimum Principle (PMP) to compute time-dependent external control fields that guide the system's time evolution toward desired states exhibiting maximum entanglement.
We propose a cost function designed to facilitate the generation of entangled states using the concurrence measure, aiming to achieve this within the minimum possible time. The optimal control problem is formulated as follows:
\begin{equation*} 
 \begin{aligned}
  &\mathop {\min }\limits_{\rho_A ,{t_f}} J = -C(\rho( t_f))+\Gamma\, {{t}_{f}} \\ 
  &\text {subject to:}  \\ 
  &\dot{\rho}_A(t) = F(\rho(t), u(t)), &\text{ (dynamics)}\\
  &\rho(t_0) = \rho_0 \in \mathcal{H}_\rho, &\text{ (initial condition)}\\ 
  &u\left( t \right) \in \mathcal{U} := \left\{ u \in {{L}_{\infty }^m}: {{\left\| u \right\|}_{\infty}} \le {u}_{\max}\right\} &\text{ (control constraint)} \\
  &t_f\in(t_0, \infty) &\text{ (final time constraint)}
\end{aligned} 
\end{equation*}
where \(\Gamma\) is a positive coefficient that balances the minimization of the time required and the concurrence. The final time \(t_f\) is treated as free and optimized. The control fields \(u(t)\) are modeled as time-dependent bounded measurable functions.

The time evolution of the density operator \({\rho}(t)\) is modeled using the Liouville-von Neumann equation
\begin{equation}\label{liou}
\dot{\rho}(t)= -i[H_0+H_c(t), \rho(t)], \quad H_c(t)= \sum _ {k=1} ^ m u_k(t)H_k    
\end{equation}
where \( H_0 \) represents the system's intrinsic Hamiltonian, 
\( H_k \) is the fixed control Hamiltonian, and \( u_k(t) \) is the time-dependent scalar control input. Note that the dynamics in the optimal control formulation is considered for the reduced density matrix $\rho_A$ with the action of the partial trace:
\[
F(\rho(t),u(t))= \operatorname{Tr}_B\dot{\rho},
\]
where the partial trace operation \( \operatorname{Tr}_B(\cdot) \) is performed over subsystem \( B \), leaving the reduced density matrix \( \rho_A(t) \) of subsystem \( A \).
In this framework, the control field \( u(t) \) is optimized to guide \( \rho_A(t) \) toward a state that maximizes the quantum property of interest, concurrence, while adhering to the system constraints and dynamics.

The existence of an optimal solution to this control problem is guaranteed by the application of PMP, under the 
regularity conditions (smoothness of dynamics, controllability, compactness of the control space and the continuity of the cost functional). 
The objective is to determine the pair \( (\rho^*, u^*) \) that is optimal in the sense that the value of the cost functional is minimized over the set of all feasible solutions, achieving the desired entanglement in minimum time.

\subsection{Necessary Optimality Conditions in the Form of Pontryagin's Minimum Principle}
In this subsection, we explore the solution to the proposed optimal control problem, which focuses on minimizing the time required to achieve high degrees of entanglement in a bipartite quantum system. This problem is subject to bounded control inputs and aims to maximize the concurrence. The solution involves determining the optimal state trajectory \( \rho^* \), the optimal control \( u^* \), and the corresponding Lagrange multiplier \( \pi^* \). These components are derived using PMP, which provides the necessary optimality conditions.

\begin{proposition}
According to Pontryagin's Minimum Principle, the optimal control law \( u^* \), corresponding to the proposed optimal control problem, exhibits a bang-bang nature, where the control \( u^*(t) \) alternates between its extreme values.
\end{proposition}

\begin{proof}
To determine \( \rho^* \), \( u^* \), and \( \pi^* \), the Pontryagin Hamiltonian function \( \mathcal{H} \) is defined as
\begin{equation}
\mathcal{H}(\rho(t), u(t), \pi(t)) = \operatorname{Tr}\left(\pi^\dagger(t) \dot{\rho}_A \right), \quad t \in [t_0, t_f]    
\end{equation}
According to Pontryagin's Minimum Principle, the optimal control \( u^*(t) \) must minimize the Hamiltonian \( \mathcal{H} \) at every time \( t \). To obtain the optimal control law, we further analyze the Pontryagin Hamiltonian as
\begin{equation*}
\mathcal{H}(\rho(t), u(t), \pi(t)) = \operatorname{Tr}\left(\pi^\dagger(t) \dot{\rho}_A \right)=-i \operatorname{Tr}\left(\pi^\dagger(t) \operatorname{Tr}_B ([H,\rho])
 \right)
\end{equation*}
where $H=H_0+H_C$. Since the quantum mechanical Hamiltonian $H$ and the density matrix $\rho$ are both Hermitian, one concludes that
$([H, \rho])^\dagger = - [H, \rho]$. This means that the commutator $[H, \rho]$ results is a skew-Hermitian matrix, so, its diagonal elements are purely imaginary. Hence, its diagonal and offdiagonals elements satisfy 
$([H, \rho])_{jj}=-([H, \rho])_{jj}^*$ and $([H, \rho])_{jk}=-([H, \rho])_{jk}^*$, respectively. 
In addition, according to trace properties, $\operatorname{Tr}([H, \rho])=\operatorname{Tr}(H\rho)-\operatorname{Tr}(\rho H)=0$, which means that the trace of $\operatorname{Tr}_B ([H,\rho])$ is also zero. Hence, we obtain the following form for $\beta \in \mathbb{C}$ and $\alpha \in i\mathbb{R}$,
\begin{equation*}
\operatorname{Tr}_B ([H,\rho])=\begin{pmatrix} \alpha & \beta \\ -\beta^* & -\alpha \end{pmatrix}
\end{equation*}
so the Pontryagin Hamiltonian can be written as
\begin{equation*}
\begin{aligned}
\mathcal{H}(\rho(t), u(t), \pi(t)) =&  -i \operatorname{Tr} \big (   \begin{pmatrix} \pi_{11} & \pi_{12} \\ \pi_{12}^* & \pi_{22} \end{pmatrix}  \begin{pmatrix} \alpha & \beta \\ -\beta^* & -\alpha \end{pmatrix} \big)\\
=& -i (\pi_{11} \alpha - \pi_{12} \beta^* + \pi_{12}^* \beta - \pi_{22} \alpha) \\ 
=&-i( (\pi_{11}-\pi_{22}) \alpha + (\pi_{12}^* \beta-(\pi_{12}^* \beta)^*))
\end{aligned}
\end{equation*}
since $\pi_{11},\pi_{11} \in \mathbb{R}$ and $(\pi_{12}^* \beta-(\pi_{12}^* \beta)^*)$ also gives a pure imaginary number, the Pontryagin Hamiltonian results in a real number. Now, we need to consider minimizing the Pontryagin Hamiltonian since
the optimal control is determined by
\[
u^*(t) = \argmin_{u(t)\in\mathcal{U}} \mathcal{H}(\rho(t), u(t), \pi(t)),
\]
for all \( t \in [t_0, t_f] \), subject to \( u(t) \in \mathcal{U} \), the admissible set of controls.
The minimization depends on the control-dependent term in \( \mathcal{H} \), i.e.,
$\operatorname{Tr}\left(\pi^\dagger(t) \operatorname{Tr}_B([H_c, \rho(t)])\right)$.
Thus, the optimal control law \( u^*(t) \) is determined by the sign of \( \operatorname{Tr} \sum_k (\pi^\dagger(t) \operatorname{Tr}_B([H_k, \rho(t)])) \), and is given as
\begin{equation}\label{control}
u^*_k(t) =
\begin{cases}
u_{k,\max}, & \text{if } -i \operatorname{Tr}(\pi^\dagger(t) \operatorname{Tr}_B([H_k, \rho(t)])) < 0, \\
-u_{k,\max}, & \text{otherwise}.
\end{cases}
\end{equation}
This control law explicitly incorporates the partial trace \( \operatorname{Tr}_B \), which accounts for the dynamics of the reduced subsystem \( A \) by tracing out the influence of subsystem \( B \). The linear dependence of \( \mathcal{H} \) on \( u(t) \), specifically through the term \( -iu_k(t)\operatorname{Tr}\left(\pi^\dagger(t) \operatorname{Tr}_B([H_k, \rho(t)])\right) \), ensures that the optimal solution lies on the boundaries of the admissible control set, leading to a bang-bang nature of \( u^*(t) \).
This bang-bang control strategy, as described in \eqref{control}, drives the system efficiently toward the desired entangled state while satisfying the constraints of the control problem.
\end{proof}
\begin{proposition}
Under PMP, the evolution of the adjoint variable \( \pi(t) \) is governed by the adjoint equation
\begin{equation}\label{adj}
 \dot{\pi}(t) = -\pi^\dagger {\dot{\rho}_A}^\prime,   
\end{equation}
where \( {\dot{\rho}_A}^\prime = \frac{\partial \dot{\rho}_A}{\partial \rho_A}\) represents the functional derivative of the time evolution operator \( \dot{\rho}_A \) with respect to the state \( \rho_A \).
At the final time \( t_f \), the adjoint variable satisfies the terminal condition
\begin{equation}\label{adjfinal}
\pi(t_f) = \frac{\sqrt{2} \rho_A(t_f)}{\sqrt{1 - \operatorname{Tr}(\rho_A(t_f)^2)}}.
\end{equation}
\end{proposition}
\begin{proof}
Under PMP, the adjoint variable \( \pi(t) \) evolves according to the adjoint equation
\[
\dot{\pi}(t) = -\frac{\partial \mathcal{H}}{\partial \rho_A(t)}=-\frac{\partial}{\partial \rho_A} \text{Tr}(\pi^\dagger(t) \dot{\rho}_A(t))=-\pi^\dagger \frac{\partial {\dot{\rho}_A(t)}}{ \partial {\rho}_A},
\]
where the Hamiltonian depends on \( \rho_A(t) \) implicitly through \( \dot{\rho}_A(t) \). 

The terminal condition for \( \pi(t) \) follows from PMP and the specific structure of the problem. At the final time \( t_f \), the adjoint variable must satisfy
\[
\pi(t_f) = \frac{\partial J}{\partial \rho_A}=-\frac{\partial C(\rho)}{\partial \rho_A}=
-\sqrt{2} \frac{1}{2 \sqrt{1 - \text{Tr}(\rho_A^2)}} (-2 \rho_A),
\]
which leads to \eqref{adjfinal} for the final time.
\end{proof}

\begin{remark}
The explicit expression for the functional derivative of \( \dot{\rho}_A \) with respect to \( \rho_A \) depends on the form of $\rho$ and the correlation term between the two systems $A$ and $B$. To analyze this, we first implement the time derivative of $\rho$ using the Schmidt decomposition as
\begin{equation}
\begin{aligned}
\dot{\rho}_A&=\operatorname{Tr}_B (-i [H_0 + H_c , \rho])=-i \operatorname{Tr}_B ([H_0 + H_c , \rho])\\
&= -i \sum_{k=1}^r \left( \mathbb{I}_A \otimes \langle v_k^B | \right) [H_0 + H_c , \rho] \left( \mathbb{I}_A \otimes |v_k^B \rangle \right)
\end{aligned} 
\end{equation}
We now consider the derivative of $\dot{\rho}_A$ with respect to ${\rho}_A $, referred to as $\dot{\rho}_A^\prime$, so
\begin{equation}
\dot{\rho}_A^\prime= -i \sum_{k=1}^r \left( \mathbb{I}_A \otimes \langle v_k^B | \right) \frac{\partial [H_0 + H_c , \rho]}{\partial \rho_A}
\left( \mathbb{I}_A \otimes |v_k^B \rangle \right) 
\end{equation}
Given the total density matrix $\rho = \rho_A \otimes \rho_B + \rho_{cor}$, 
where \( \rho_{\mathrm{cor}} \) represents the correlation term. We analyze two cases:

\begin{itemize}
    \item[-] \text{Case 1: \( \rho_{\mathrm{cor}} = 0 \) (No Correlations Between \( A \) and \( B \))}

    If there are no correlations, the total density matrix \( \rho \) is separable:
    \[
    \rho = \rho_A \otimes \rho_B.
    \]
    In this case, the derivative of \( \rho \) with respect to \( \rho_A \) is straightforward. Since \( \rho_B \) is independent of \( \rho_A \), we have
    \[
    \frac{\partial \rho}{\partial \rho_A} = \mathbb{I}_B \otimes \rho_B.
    \]
    Substituting this into the functional derivative of \( \dot{\rho}_A \), we obtain
    \[
    \dot{\rho}_A' = -i \sum_{k=1}^r \left( \mathbb{I}_A \otimes \langle v_k^B | \right) \left[ H_0 + H_c, \mathbb{I}_B \otimes \rho_B \right] \left( \mathbb{I}_A \otimes |v_k^B \rangle \right).
    \]
    Here, the commutator simplifies because of the separability of \( \rho \).

    \item[-] \text{Case 2: \( \rho_{\mathrm{cor}} \neq 0 \) (Presence of Correlations)}

    When correlations exist, the total density matrix \( \rho \) includes the correlation term \( \rho_{\mathrm{cor}} \), and the derivative becomes more complicated. Specifically, the functional derivative of \( \rho \) with respect to \( \rho_A \) is given by
    \[
    \frac{\partial \rho}{\partial \rho_A} = \mathbb{I}_B \otimes \rho_B + \frac{\partial \rho_{\mathrm{cor}}}{\partial \rho_A}.
    \]
    The second term \( \frac{\partial \rho_{\mathrm{cor}}}{\partial \rho_A} \) accounts for the interdependence between subsystems \( A \) and \( B \) through their correlations. Substituting this into the expression for \( \dot{\rho}_A' \), we have
    \begin{equation*}
        \begin{aligned}
             \dot{\rho}_A'  & = -i \sum_{k=1}^r \left( \mathbb{I}_A \otimes \langle v_k^B | \right) \left[ H_0 + H_c, \mathbb{I}_B \otimes \rho_B + \frac{\partial \rho_{\mathrm{cor}}}{\partial \rho_A} \right] \\
            &\left( \mathbb{I}_A \otimes |v_k^B \rangle \right)
        \end{aligned}
    \end{equation*}
    Here, the correlation term \( \rho_{\mathrm{cor}} \) introduces additional complexity, and its specific form determines the nature of the functional derivative.
\end{itemize}
\end{remark}

\begin{proposition}
Given the cost functional \( J \) and the Pontryagin Hamiltonian \( \mathcal{H}(t) \), for the Pontryagin Hamiltonian associated with the optimization problem, at the final time \( t_f \), the following transversality condition holds:
\begin{equation}\label{htf}
\mathcal{H}(t_f) =\frac{\sqrt{2} \rho_A }{2 \sqrt{1 - \text{Tr}(\rho_A^2)}} \dot{\rho}_A - \Gamma
\end{equation}
\end{proposition}
\begin{proof}
From the transversality condition of the final time one has
\begin{equation*}
\mathcal{H}(t_f) = -\frac{\partial J}{\partial t_f} =-\frac{\partial (-C(\rho( t_f))+\Gamma t_f) }{\partial t_f}.
\end{equation*}
We need to compute 
\begin{equation*}
\begin{aligned}
\frac{\partial C(\rho(t_f))}{\partial t_f} 
 &= \frac{\partial C(\rho)}{\partial \rho_A(t_f)}
 \frac{\partial \rho_A(t_f)}{\partial t_f}=\frac{\sqrt{2} \rho_A }{2 \sqrt{1 - \text{Tr}(\rho_A^2)}} \dot{\rho}_A
\end{aligned}
\end{equation*}
In addition $-\frac{\partial \Gamma t_f}{\partial t_f}=- \Gamma$, so by summing the two terms one obtains \eqref{htf}.
\end{proof}

\section{Numerical Validation of Optimal Entanglement Control}\label{sec:numerical_validation}
In this section, we numerically validate the entanglement control strategy developed in the preceding section. We begin by introducing the system Hamiltonian. The Pauli matrices, which form the basis for constructing the system's Hamiltonians, are defined as
\[
\sigma_x = \begin{pmatrix} 0 & 1 \\ 1 & 0 \end{pmatrix}, \quad 
\sigma_y = \begin{pmatrix} 0 & -i \\ i & 0 \end{pmatrix}, \quad 
\sigma_z = \begin{pmatrix} 1 & 0 \\ 0 & -1 \end{pmatrix}.
\]
In the proposed example, the internal Hamiltonian \( H_0 \) (in the absence of external fields) is given by
\[
H_0 = \sigma_z \otimes \sigma_z.
\]
The control Hamiltonian \( H_c(t) \), driven by time-dependent control fields \( u_k(t) \), consists of a linear combination of coupling Hamiltonians \( H_k \), where \( k = 1, 2, 3 \), constructed as
\[
\begin{aligned}
    H_1 &= \sigma_x \otimes \sigma_y + \sigma_z \otimes \sigma_z, \\
    H_2 &= \sigma_x \otimes \sigma_z + \sigma_z \otimes \sigma_x, \\
    H_3 &= \sigma_y \otimes \sigma_z + \sigma_z \otimes \sigma_y.
\end{aligned}
\]

We consider a general two-qubit superposition state 
as
\[
|\psi(t)\rangle = \alpha_{00}|00\rangle + \alpha_{01}|01\rangle + \alpha_{10}|10\rangle + \alpha_{11}|11\rangle,
\]
where \( \alpha_{ij} \) are complex coefficients. The corresponding density matrix is given by
\[
\rho = 
\begin{bmatrix}
|\alpha_{00}|^2 & \alpha_{00}\alpha_{01}^* & \alpha_{00}\alpha_{10}^* & \alpha_{00}\alpha_{11}^* \\
\alpha_{01}\alpha_{00}^* & |\alpha_{01}|^2 & \alpha_{01}\alpha_{10}^* & \alpha_{01}\alpha_{11}^* \\
\alpha_{10}\alpha_{00}^* & \alpha_{10}\alpha_{01}^* & |\alpha_{10}|^2 & \alpha_{10}\alpha_{11}^* \\
\alpha_{11}\alpha_{00}^* & \alpha_{11}\alpha_{01}^* & \alpha_{11}\alpha_{10}^* & |\alpha_{11}|^2
\end{bmatrix}.
\]
To quantify entanglement, we compute the partial trace over subsystem \( B \), resulting in the reduced density matrix:
\[
\text{Tr}_B(\rho) = 
\begin{bmatrix}
|\alpha_{00}|^2 + |\alpha_{01}|^2 & \alpha_{00} \alpha_{10}^* + \alpha_{01} \alpha_{11}^* \\
\alpha_{10} \alpha_{00}^* + \alpha_{11} \alpha_{01}^* & |\alpha_{10}|^2 + |\alpha_{11}|^2
\end{bmatrix}.
\]
We focus on driving the system toward Bell states, which represent maximally entangled two-qubit states. The Bell states are defined as
\[
\begin{aligned}
    |\Phi^+\rangle &= \frac{1}{\sqrt{2}}(|00\rangle + |11\rangle),\quad |\Phi^-\rangle &= \frac{1}{\sqrt{2}}(|00\rangle - |11\rangle) \\ 
    |\Psi^+\rangle &= \frac{1}{\sqrt{2}}(|01\rangle + |10\rangle), \quad |\Psi^-\rangle &= \frac{1}{\sqrt{2}}(|01\rangle - |10\rangle).   
\end{aligned}
\]
These states exhibit specific entanglement patterns and serve as benchmarks for evaluating the effectiveness of our control strategy.
To investigate the dynamics of entanglement generation, we initialize the system in a perturbed separable state to avoid critical points where the entanglement gradient vanishes. The initial state takes the form
\[
\rho_0 = (1 - \epsilon)\rho_{\text{sep}} + \epsilon \delta\rho,
\]
where \( \rho_{\text{sep}}\) represents a separable state, \( \delta\rho \) introduces a small entangled component, and \( 0<\epsilon \ll 1 \) controls the perturbation magnitude. This setup ensures the system starts close to separable while still enabling smooth optimization dynamics. By applying the control fields \( u_k(t) \), we numerically optimize the state evolution to achieve a maximally entangled target state.
For instance, we considered a perturbed separable state by expressing $\rho_{\text{sep}}=|00\rangle\langle 00|$ and $\delta\rho=|\Phi^+\rangle\langle \Phi^+|$.

Figure \ref{fig:control_fields} illustrates the evolution of the control fields \(u_1\), \(u_2\), and \(u_3\) over time. Notably, the control field \(u_1\) exhibits switching behavior, transitioning between \(-1\) and \(1\) at specific times. This behavior is interpreted as a bang-bang control strategy.
The observed switching aligns with the principles of the PMP, where the control input is driven to its extremal values to maximize system performance.
The abrupt transitions (e.g., around \(t = 0.5\)) correspond to changes in the system's switching function, which determines the optimal control direction. The fields \(u_2\) and \(u_3\) remain constant at the lower bound.

Figure \ref{fig:optimal_concurrence} illustrates the evolution of concurrence over time under the influence of optimal switching control fields:
The figure demonstrates the effectiveness of optimal control in generating entanglement in a quantum system. Key observations include:
At the beginning (\(t < 0.3\)), the concurrence increases steadily as the control fields drive the system from a separable state toward an entangled state. After \(t = 0.6\), the concurrence approaches its maximum value of 1, representing a maximally entangled state. The rapid transitions in the control fields, as seen in the previous figure, optimize the state evolution to maximize entanglement efficiently. Such behavior is essential for quantum technologies, including quantum communication and quantum computing, where high levels of entanglement are crucial.

\begin{figure}[t]
    \centering
    \includegraphics[width=0.47\textwidth]{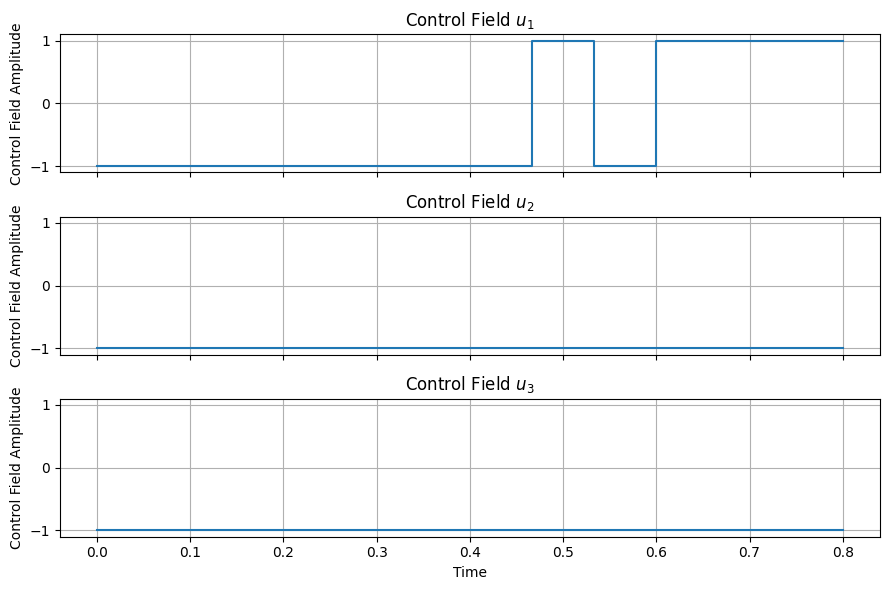}
    \caption{Control field amplitudes \(u_1\), \(u_2\), and \(u_3\) as a function of time.}
    \label{fig:control_fields}
\end{figure}

\begin{figure}[t]
    \centering
    \includegraphics[width=0.47\textwidth]{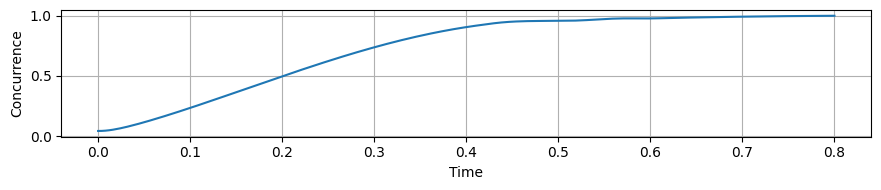} 
    \caption{Evolution of concurrence as a function of time using optimal switching control fields. The switching behavior in the control fields ensures a rapid increase in concurrence, reaching its maximum value of 1, indicative of maximal entanglement.}
    \label{fig:optimal_concurrence}
\end{figure}

\section{Conclusions}
This work demonstrates the efficacy of Pontryagin's Minimum Principle in designing optimal control strategies for generating maximally entangled states in bipartite quantum systems. By formulating and solving a control problem centered on maximizing concurrence, the proposed approach efficiently guides the system dynamics using time-dependent switching control fields. Numerical validation confirms the rapid entanglement generation and achieving to the maximal concurrence. Future work could extend this framework to open quantum systems, multi-partite entanglement, and the implications of this method for advancing quantum technologies that rely on high-fidelity entanglement, such as quantum communication, sensing, and computation. 

\bibliographystyle{IEEEtran}

\bibliography{ref}

\end{document}